\documentclass[graybox]{svmult}

\usepackage{latexsym,amsfonts,amsmath,theorem,amssymb}

\usepackage{mathptmx}       
\usepackage{helvet}         
\usepackage{courier}        
\usepackage{type1cm}        
                            
\usepackage{makeidx}         
\usepackage{graphicx}        
\usepackage{multicol}        
\usepackage[bottom]{footmisc}

\usepackage{url}

\DeclareFixedFont{\ttb}{T1}{txtt}{bx}{n}{8}
\DeclareFixedFont{\ttm}{T1}{txtt}{m}{n}{8}
\DeclareFixedFont{\ttc}{T1}{txtt}{m}{n}{8}
\usepackage{listings}
\lstdefinestyle{pythonstyle}{
language=Python,
basicstyle=\ttm,
otherkeywords={self},
keywordstyle=\ttb\color{black},
frame=tb,
commentstyle=\color{gray}\ttc,
showstringspaces=false,
breaklines=true,
postbreak=\mbox{\textcolor{gray}{$\hookrightarrow$}\space}
}
\usepackage{tikz}
\usetikzlibrary{decorations.pathreplacing,angles,quotes}
\usetikzlibrary{arrows,calc,patterns,angles,quotes}

\usepackage{todonotes}
\graphicspath{{figures/}}

\newcommand{\is}{\text{IS}}
\newcommand{\ess}{M_{\text{eff}}}
\newcommand{\mlow}{M_{\text{thresh}}}

\makeindex             

\begin{document}

\title*{A practical example for the non-linear Bayesian filtering of model parameters}
\titlerunning{Non-linear Bayesian filtering}
\author{Matthieu Bult\'e, Jonas Latz, Elisabeth Ullmann}
\institute{Matthieu Bult\'e, Jonas Latz, Elisabeth Ullmann \at Zentrum Mathematik, Technische Universit\"at M\"unchen, Boltzmannstra{\ss}e 3, 85748 Garching b.M., Germany, \email{matthieu.bulte@tum.de, jonas.latz@ma.tum.de, elisabeth.ullmann@ma.tum.de}}
%
%
\maketitle

\abstract{
In this tutorial we consider the non-linear Bayesian filtering of static parameters in a time-dependent model.
We outline the theoretical background and discuss appropriate solvers.
We focus on particle-based filters and present Sequential Importance Sampling (SIS) and Sequential Monte Carlo (SMC).
Throughout the paper we illustrate the concepts and techniques with a practical example using real-world data.
The task is to estimate the gravitational acceleration of the Earth $g$ by using observations collected from a simple pendulum.
Importantly, the particle filters enable the adaptive updating of the estimate for $g$ as new observations become available.
For tutorial purposes we provide the data set and a Python implementation of the particle filters.
}

\section{Introduction}

An important building block of uncertainty quantification is the statistical estimation and sustained learning of parameters in mathematical and computational models.
In science and engineering models are used to emulate, predict, and optimise the behaviour of a system of interest.
Examples include the transport of contaminants by groundwater flow in hydrology, the price of a European option in finance, or the motion of planets by mutual gravitational forces in astrophysics. 
The associated mathematical models for these examples are an elliptic partial differential equation (PDE), the parabolic Black-Scholes PDE, and a system of ordinary differential equations (ODEs) describing the $N$-body dynamics.

Assuming that we have observational data of the system of interest, it is now necessary to calibrate the model with respect to these observations.
This means that we identify model parameters such that the model output is close to the observations in a suitable metric.
In the examples above we need to calibrate the hydraulic conductivity of the groundwater reservoir, the volatility of the stock associated with the option, and the masses of the planets.

In this tutorial  we focus on the next step following the model calibration, namely the updating of the estimated parameters as additional observations become available.
This is an important task since many systems are only partially observed. 
Thus it is often unlikely to obtain high quality estimates of underlying model parameters by using only a single data set.
Moreover, it is often very expensive or impossible  to restart the parameter estimation with all data sets after a new data set becomes available.
The problem of combining a parameter estimate with a new set of observations to update the estimate based on all observations is called \emph{filtering} in statistics.
Filtering can be considered as a \emph{learning process}: a certain state of knowledge based on previous observations is combined with new observations to reach an improved state of knowledge.

Throughout this tutorial we consider a practical example for filtering. 
We study the periodic motion of a pendulum.
The underlying mathematical model is an ODE.
The model parameters are the length of the pendulum string $\ell$, and the gravitational acceleration of the Earth $g$.
We assume that $\ell$ is known, however, we are uncertain about $g$. 
Our goal is to estimate and update the estimate for $g$ based on real-world observational data.
Importantly, the pendulum experiment can be carried out without expensive equipment or time-consuming preparations.
Moreover, the mathematical model is simple and does not require sophisticated or expensive numerical solvers.
However, the filtering problem is non-linear and non-Gaussian.
It does not have an analytic solution, and an efficient approximate solution must be constructed. 
We use particle filters for this task. 

The simple pendulum setting allows us to focus on the statistical aspects of the estimation problem and the construction of particle filters.
The filters we discuss are well known in the statistics and control theory communities, and textbooks and tutorials are available, see e.g. \cite{Doucet2011, Law2015, Sarkka2013}.
However, these works focus on filters for \emph{state space estimation}.
In contrast, we employ filters for \emph{parameter estimation} in mathematical models, and within the Bayesian framework.

Bayesian inverse problems attracted a lot of attention in the applied mathematics community during the past decade since the work by Stuart \cite{Stuart2010} which laid out the mathematical foundations of Bayesian inverse problems.
The design of efficient solvers for these problems is an active area of research, and particle filters offer attractive features which deserve further research.
This tutorial enables interested readers to learn the building blocks of particle filters illustrated by a simple example.
Moreover, we provide the source code so that the reader can combine the filters with more sophisticated mathematical models.

The remaining part of this tutorial is organised as follows. 
In \S \ref{Subs_Problem_Form} we give a precise formulation of the filtering problem and define a filter.
We introduce the pendulum problem in \S \ref{Sec_Pendulum}, and review previous work on model calibration, filtering and the numerical approximation of these procedures in \S \ref{Sub_SoA}.
In \S \ref{Sec_Bayes} we introduce the Bayesian solution to the filtering problem. 
Furthermore, we explain the statistical modeling of the pendulum filtering problem.
In \S \ref{Sec_SMC} we discuss particle-based filters, namely Sequential Importance Sampling and Sequential Monte Carlo.
In \S \ref{Subs_Num_Pendulum} we apply both these methods to the pendulum filtering problem, and comment on the estimation results.
Finally, we provide a discussion in \S \ref{sec:discussion}.

\section{The filtering problem}
\label{Subs_Problem_Form}
We motivated the filtering of model parameters in the preceding section. 
Now we give a rigorous introduction to filtering.
Note that we first define the filtering problem in a general setting. 

Let $X$ and $H$ be separable Banach spaces.
$X$ denotes the \emph{parameter space}, and $H$ denotes the \emph{model output space}.
We define a \emph{mathematical model} $G: X \rightarrow H$ as a mapping from the parameter space to the model output space. 
Next, we observe the system of interest that is represented by the model.
We collect measurements at $T \in \mathbb{N}$ points in time $t = 1,\dots,T$. 
These observations are denoted by $y_1, y_2,\dots,y_T$. 
Each observation $y_t$ is an element of a \emph{finite-dimensional} Banach space $Y_t$.
The family of spaces $Y_1,\dots,Y_T$ are the so called \emph{data spaces}.
We model the observations by \emph{observation operators} $\mathcal{O}_t: H \rightarrow Y_t$, $t \geq 1$, that map the model output to the associated observation. 
Furthermore, we define a family of \emph{forward response operators} $\mathcal{G}_t := \mathcal{O}_t \circ G$, $t \geq 1$, that map from the parameter space directly to the associated data space.
We assume that the observations are noisy and model this fact by randomness.
The randomness is represented on an underlying probability space $(\Omega, \mathcal{A}, \mathbb{P})$.
Each observation $y_t$ is the realisation of a random variable $\widetilde{y}_t: \Omega \rightarrow Y_t$.
Moreover,
\begin{equation}
\widetilde{y}_t \sim L_t(\cdot|\theta^\dagger)
\end{equation}
where $L_t: Y_t \times X \rightarrow [0, \infty)$, $t \geq 1$, is a parameterised probability density function (w.r.t. the Lebesgue measure). 
$\theta^\dagger$ denotes the true parameter associated with the observations.

\begin{example}[Additive Gaussian noise] \label{Example_Gaussian_noise_add}
A typical assumption is that the measurement noise is Gaussian and additive. 
In that case $y_t$ is a realisation of the random variable
$\widetilde{y}_t = \mathcal{G}_t(\theta^\dagger) + \eta_t$,
where $\eta_t \sim \mathrm{N}(0,\Gamma_t)$ and $\Gamma_t: Y_t \rightarrow Y_t$ is a linear, symmetric, positive definite covariance operator, $t \geq 1$.
It holds
$$L_t(y_t|\theta) 
= \exp\left(-\frac{1}{2}\|\Gamma^{-1/2}(\mathcal{G}_t(\theta)-y_t)\|^2 \right).$$
\end{example}

The \emph{inverse} or \emph{smoothing problem} at a specific timepoint $t \geq 1$ is the task to  identify the unknown true parameter $\theta^\dagger$ given the data set $(y_1,\dots,y_t) =: y_{1:t}$.
We denote the estimate for the parameter by $\widehat{\theta}(y_{1:t})$. 
Hence, a formal expression for smoothing is the map $$y_{1:t} \mapsto \widehat{\theta}(y_{1:t}).$$
The \emph{filtering problem}, on the other hand, is the task to update the estimate $\widehat{\theta}(y_{1:t})$ after the observation $y_{t+1}$ is available. 
Hence, a formal expression for a \emph{filter} is the map $$\{\widehat{\theta}(y_{1:t}), y_{t+1}\}  \mapsto \widehat{\theta}(y_{1:t+1}).$$
Filtering can be considered as a \emph{learning process} in the following sense.
Our point of departure is a current state of knowledge represented by the parameter estimate $\widehat{\theta}(y_{1:t})$.
This involves all observations up to the point in time $t$.
The data set $y_{t+1}$ is then used to arrive at a new state of knowledge represented by the updated parameter estimate $\widehat{\theta}(y_{1:t+1})$.
We depict this learning process in Fig. \ref{Fig_tikz_learning}. 
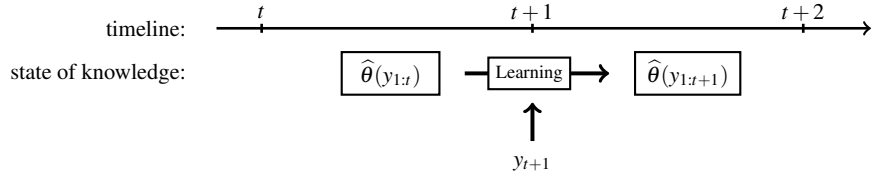
\begin{figure}[htb]
\centering
\begin{tikzpicture}[scale=0.6]
\draw	(1.5,0)node[anchor=east] {timeline:};

 \draw	(1.5,-1)node[anchor=east] {state of knowledge:};
 
 \draw	(4.75,-1)node[rectangle,anchor=west,text width=1.1cm, thick,draw=black,align=center,fill=white] {$\widehat{\theta}(y_{1:t})$};
 
  \draw	(3,0.35)node[align=center] {$t$};

    \draw	(15,0.35)node[align=center] {$t+2$};
  \draw[line width = 1.0, color = black,-](3,0.1) -- (3,-0.1);
    \draw[line width = 1.0, color = black,-](9,0.1) -- (9,-0.1);
    \draw[line width = 1.0, color = black,-](15,0.1) -- (15,-0.1);

  \draw[line width = 1.5, color = black,->](7.5,-1) -- (10.7,-1);
    \draw[line width = 1.0, color = black,->](2,0) -- (16.5,0);

    \draw[line width = 1.5, color = black,<-](9,-1.65) -- (9,-2.5);
      \draw	(9,-3)node[align=center] {$y_{t+1}$};
\draw	(8,-1)node[align=center,rectangle,anchor=west,thick,draw=black,fill=white] {\scriptsize Learning};
   \draw	(11.25,-1)node[align=center,rectangle,anchor=west,text width=1.2cm,thick,draw=black,fill=white] {$\widehat{\theta}(y_{1:t+1})$};
      \draw	(9,0.35)node[align=center] {$t+1$};
 \end{tikzpicture}
 \caption{The filtering problem. The starting point is the current state of knowledge $\widehat{\theta}(y_{1:t})$ at the point in time $t$. At the timepoint $t+1$ we observe $y_{t+1}$. We want to use these observations to improve our knowledge concerning $\theta^\dagger$. The new state of knowledge is given by the updated estimate $\widehat{\theta}(y_{1:t+1})$. }
\label{Fig_tikz_learning}
\end{figure}

Next we describe two practical filtering problems.
The pendulum filtering problem is used for illustration purposes, and the tumor filtering problem highlights a more involved application of filtering.

\begin{example}[Tumor]
The tumor inverse and filtering problem has been discussed extensively in the literature, see e.g. \cite{Collis2017, Kahle2018, Lima2017} and the references therein.
In this problem we model a tumor with a system of (partial) differential equations, for example, the Cahn-Hilliard or reaction-diffusion equations, or, alternatively, an atomistic model.
The goal is to predict the future growth of the tumor.
Moreover, we wish to test, compare and select suitable therapeutical treatments.
To do this we need to estimate model parameters, e.g. the tumor proliferation and consumption rate, and chemotaxis parameters.
These model parameters are patient-specific and can be calibrated and updated using patient data.
The data is given by tumor images obtained e.g. with \emph{magnetic resonance imaging} (MRI) or with \emph{positron emission tomography} (PET).
The images are captured at different timepoints and monitor the progression of the tumor growth.
Note that the data spaces are in general infinite-dimensional in this setting.
\end{example}

\subsection{Pendulum example} \label{Sec_Pendulum}

In this section we describe a simple yet practical filtering problem that is associated with a real-world experiment and data.
Throughout this tutorial we will come back to this problem to illustrate the filtering of model parameters.

The goal of the pendulum inverse problem is the estimation of the Earth's gravitational acceleration $g$ 
using measurements collected from the periodic motion of a pendulum. 
Note that the gravitational acceleration at a particular location depends on the altitude and the latitude of this location. 
We use measurements that were collected in Garching near Munich, Germany, 
where the height above mean sea level is $h = 482m$ and the latitude is $\phi = 48^\circ 15' \ \mathrm{N} = 48.25^\circ \  \mathrm{N}$.
The formula (4.2) in \cite{Allmaras2013} gives the gravitational acceleration in Garching:
\begin{align*}
g^\dagger = &\Big(9.780327 \left(1 + 5.3024\cdot 10^{-3} \sin^2(\phi) - 5.8\cdot 10^{-6}\sin^2(2\phi)\right)\\&-1.965 \cdot 10^{-6}h m^{-1} \Big)\frac{m}{s^2} \approx 9.808 \frac{m}{s^2}.
\end{align*}

We use a simplified model to describe the dynamics of the pendulum.
Specifically, we ignore friction, and assume that the pendulum movements take place in a two-dimensional, vertically oriented plane. 
In this case the state of the pendulum can be described by a single scalar that is equal to the angle enclosed by the pendulum string in its excited position and the stable equilibrium position.
By using Newton's second law of motion and by considering the forces acting on the pendulum it is easy to see that this angle $x(\cdot; g)$ satisfies the parametrised non-linear initial value problem (IVP)
\begin{align*}
  \ddot{x}(\tau; g) &= -\frac{g}{\ell}\sin(x(\tau; g)),\\
  \dot{x}(0; g) &= v_0,\\
  x(0; g) &= x_0,
\end{align*}
where $\ell$ denotes the length of the string that connects the two ends of the pendulum, and $\tau \in [0, \infty)$ denotes time. 
An illustration of the model is given in Fig. \ref{Fig_tikz_Pendulum}.

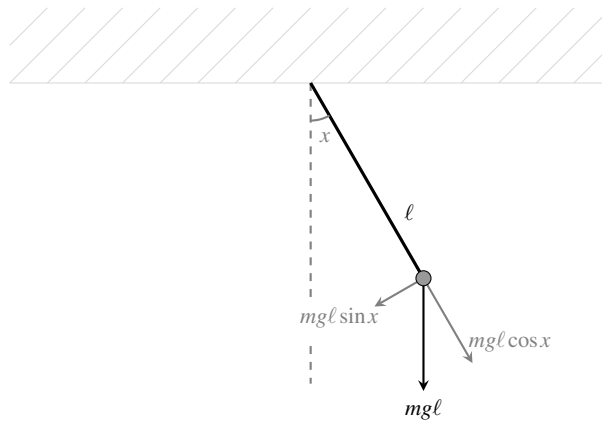
\begin{figure}[h]
  \centering
  \begin{tikzpicture}
    \pgfmathsetmacro{\Gvec}{1.5}
    \pgfmathsetmacro{\myAngle}{30}
    \pgfmathsetmacro{\Gcos}{\Gvec*cos(\myAngle)}
    \pgfmathsetmacro{\Gsin}{\Gvec*sin(\myAngle)}

    \coordinate (centro) at (0,0);
    \draw[gray!30] (-4,0) -- (4,0);
    \draw[gray!30] 
    (-4,0) -- (-3,1)
    (-3,0) -- (-2,1)
    (-2,0) -- (-1,1)
    (-1,0) -- (0,1)
    (0,0) -- (1,1)
    (1,0) -- (2,1)
    (2,0) -- (3,1)
    (3,0) -- (4,1)
    
    (-4,0.5) -- (-3.5,1)
    (-3.5,0) -- (-2.5,1)
    (-2.5,0) -- (-1.5,1)
    (-1.5,0) -- (-0.5,1)
    (-0.5,0) -- (0.5,1)
    (0.5,0) -- (1.5,1)
    (1.5,0) -- (2.5,1)
    (2.5,0) -- (3.5,1)
    (3.5,0) -- (4,0.5)
    ;
    \draw[dashed,gray,-,thick] (centro) -- ++ (0,-2.9) node (mary) [black,below]{$ $};
    \draw[dashed,gray,-,thick] (0,-3.5) -- ++ (0,-0.5);
    \draw[very thick] (centro) -- ++(270+\myAngle:3) coordinate (bob)
      node[near end,above right] {$\ell$};
 \pic [gray,draw, -, "$x$", angle eccentricity=1.5,thick] {angle = mary--centro--bob};
    \draw [gray,-stealth,thick] (bob) -- ($(bob)!-\Gcos cm!(centro)$)
      coordinate (gcos)
      node[near end, right] {$mg\ell\cos x$};
    \draw [gray,-stealth,thick] (bob) -- ($(bob)!\Gsin cm!90:(centro)$)
      coordinate (gsin)
      node[near end,below left] {$mg\ell\sin x$};
    \draw [-stealth,thick] (bob) -- ++(0,-\Gvec)
      coordinate (g)
      node[below] {$mg \ell$};
    \filldraw [fill=black!40,draw=black] (bob) circle[radius=0.1];
\end{tikzpicture}

  \caption{Pendulum model and forces applied to the bob with mass $m$. 
  	The black vertical vector represents the gravitational force. It is decomposed into the gray vectors representing the components parallel and perpendicular to the motion of the pendulum. 
  	The dashed line represents the angle $x=0$ where the pendulum is at rest. In this position the time measurements are taken. 
  	{This figure is adapted from \cite{Bulte2018}.}}
  \label{Fig_tikz_Pendulum}
\end{figure}

Following the framework presented in \S \ref{Subs_Problem_Form}, we define the model $G(g) = x(\cdot; g)$ which maps the model parameter, here the gravitational acceleration, to the model output, here the time-dependent angle.
For $t  = 1, \dots, 10$ we define the observation operator $\mathcal{O}_t$ by $\mathcal{O}_t(x(\cdot; g)) = x(\tau_t;g)$.
This models the angle measurement of the pendulum at a fixed point in time $\tau_t$.
Note that in practise the measurement is reversed since we measure the time at a prescribed angle that is easy to identify. 
Mathematically, this can be interpreted as angle measurement at a specific timepoint.
Finally, we define the forward response operators $\mathcal{G}_t = \mathcal{O}_t \circ G$ for $t = 1, \dots, 10$.
The data set $y_1, \ldots, y_t$ of angle measurements corresponds to realisations of the random variables
\begin{equation*}
  \widetilde{y}_t := \mathcal{G}_t(g^\dagger) + \eta_t 
\end{equation*}
where $\eta_1, \ldots, \eta_t$ are independent and identically distributed Gaussian random variables according to $N(0, \sigma^2)$.
The pendulum filtering problem consists of using time measurements (and the associated angle measurements) to sequentially improve the estimate of the true value of the Earth's gravitational acceleration $g^\dagger$.

\begin{remark} \label{Remark_linearisedODE}
	It is possible to simplify the mathematical model of the pendulum motion.
	Assume that $v_0 = 0$. If $|x|$ is small, then $x \approx \sin(x)$. Hence, the nonlinear ODE $\ddot{x}=-(g/\ell) \sin(x)$ above can be replaced by the linear ODE $\ddot{x}=-(g/\ell) x$ with the analytical solution
	\[x(\tau; g) := x_0\cos(\tau\sqrt{g/\ell}).\]
	However, the relation between the angle $x$ and the model parameter $g$ is still nonlinear and thus the filtering problem remains nonlinear with no analytical solution.
	For this reason we do not consider the linear pendulum dynamics.
\end{remark}

\subsection{State of the art}\label{Sub_SoA}

Model calibration problems have often been approached with optimisation techniques, for example, the Gauss-Newton or Levenberg-Marquardt algorithm. 
These algorithms minimise a (possibly regularised) quadratic loss function which measures the distance between the data and the model output, see e.g. \cite[\S 10]{Noce06}.
Today's availability of high-performance computing resources has enabled statistical techniques for the calibration of computationally expensive models.
A popular example is Bayesian inference.
Here we consider the  model parameters to be uncertain and model the associated uncertainty with a probability measure. 
By using Bayes' formula it is possible to include information from the observational data in this probability measure. 
In particular, the probability measure is conditioned with respect to the data (see \cite{Robert2007}).

Recently, Bayesian inference for model parameters (so-called Bayesian inverse problems) has attracted a lot of attention in the literature. 
It was first proposed in \cite{Kaipio2005} and thoroughly analysed in \cite{Dashti2017,Stuart2010}.
A tutorial on Bayesian inverse problems is given by Allmaras et al. \cite{Allmaras2013}; in fact this work inspired the authors of this article.
However, most works on Bayesian inverse problems, including the works cited above, are concerned with the so-called \emph{static} Bayesian learning where one uses a single set of observations and no filtering is carried out.
Taking the step from including a single set of observations to iteratively including more observations is both practically important and non-trivial.
We mention that the filtering of model parameters is closely related to the filtering of states in state space models.
In the pendulum problem in \S \ref{Sec_Pendulum} this task corresponds to the tracking and prediction of the pendulum motion.
Filtering of states is a central problem in \emph{data assimilation} (see e.g. \cite{Doucet2011,Law2015,Nakamura2015}).

Linear-Gaussian filtering problems can be solved analytically with the Kalman filter, see e.g. \cite{Humpherys2012}.
For non-linear filtering problems several non-linear approximations to the Kalman filter have been proposed. Examples are the extended Kalman filter (EKF, see e.g. \cite{Humpherys2012}) and the Ensemble Kalman Filter (EnKF, see \cite{Evensen2003}) which are essentially based on linearisations of the forward problem.
The EKF uses a Taylor expansion for linearisation.
The EnKF uses a probabilistic linearisation technique called \emph{Bayes linear} (see e.g. \cite{Latz2016,Schillings2017}).
The Taylor expansion in the EKF can be evaluated analytically, however, for the Bayes linear approximation this is not possible.
For this reason the EnKF uses a particle approximation -- an \emph{ensemble} of unweighted particles.

Another family of approximations, so-called \emph{particle filters}, do not use  linearisation strategies, but rely on importance sampling (see \cite{Robert2004}).
We will discuss two particle filters -- Sequential Importance Sampling and Sequential Monte Carlo -- in this tutorial.
Note that these methods are not only used for filtering.
Indeed, it is possible to sample from general sequences of probability measures by using particle filters.
We refer to \cite{DelMoral2004,DelMoral2013,DelMoral2006} for the theoretical background and further applications of particle filters.
We mention the use of particle filters in static Bayesian inverse problems. 
SMC and SIS are here used together with a tempering of the likelihood (see e.g. \cite{Beskos2016, Beskos2015, Kahle2018, Kantas2014, Neal2001}), with multiple resolutions of the computational model (see e.g. \cite{Beskos2017a}) or with a combination of both (see e.g. \cite{Latz2018}).
An excellent overview of SMC and particle filters can be found on the webpage of Doucet \cite{Doucet_Webpage}.
Finally, note that SMC requires a Markov kernel that is typically given by a Markov Chain Monte Carlo (MCMC) method \cite{Liu2004, Robert2004}.

\section{Bayesian filtering} \label{Sec_Bayes}
In this section we explain the Bayesian approach to the filtering problem discussed in \S\ref{Subs_Problem_Form}.
We begin by introducing conditional probabilities and their construction using Bayes' formula.
We conclude this section by revisiting the pendulum filtering problem presented in \S\ref{Sec_Pendulum} and discuss implications from the Bayesian approach.

Consider some uncertain parameter $\theta$.
We model the uncertainty as a random variable; by definition this is a measurable map from the probability space $(\Omega, \mathcal{A}, \mathbb{P})$ to the parameter space $X$.
{The expected value $\mathbb{E}(\cdot)$ is the operator integrating a function with respect to $\mathbb{P}$.}
The probability measure of $\theta$ is called $\mu_0 = \mathbb{P}(\theta \in \cdot)$.
The measure $\mu_0$ reflects the knowledge concerning $\theta$ before we include any information given by observations. 
For this reason $\mu_0$ is called \emph{prior (measure)}.

In \S \ref{Subs_Problem_Form} we modeled the data generation at time $t$ as an event that occurs and that we observe. 
This event is $\{\widetilde{y}_t = y_t\} \in \mathcal{A}$.
The process of \emph{Bayesian learning} consists in including the information ``$\widetilde{y}_t = y_t$'' into the probability distribution of $\theta$.
This is done with conditional probability measures. 
A good intuition about conditional probabilities can be obtained by consideration of discrete probabilities.

\begin{example}[Conditional probability] \label{Example_CondProb_Disc1}
Let $\theta$ denote a uniformly distributed random variable on the parameter space $X := \{1,\dots,10\}$.
The uniform distribution models the fact that we have no information whatsoever about $\theta$. 
Next, an oracle tells us that ``$\theta$ is about $4$''.
We model this information by assuming that $\theta$ is equal to $3$, $4$ or $5$ with equal probability. 
Our state of knowledge is then modeled by the following conditional probabilities:
\begin{align*}
\mathbb{P}(\theta = k | \theta \text{ is about } 4) := \begin{cases}
1/3, &\text{if } k = 3,4,5, \\
0, &\text{otherwise.}
\end{cases}
\end{align*}
\end{example}
We revisit this example in the next subsection and illustrate the computation of the conditional probability measure.
For now, we move back to the filtering problem. 
Having observed the first data set $y_1$ we replace the prior probability measure $\mu_0$ by the conditional probability measure 
\[
\mu_1 
= 
\mathbb{P}(\theta \in \cdot| \widetilde{y}_1 = y_1).
\]
Analogously to Example \ref{Example_CondProb_Disc1} the measure $\mu_1$ now reflects the knowledge about $\theta$ given the information that the event $\{\widetilde{y}_1 = y_1\}$ occurred.
In the next step we observe $\widetilde{y}_2 = y_2$ and update $\mu_1 \mapsto \mu_2$, where
\[
\mu_2 
= 
\mathbb{P}(\theta \in \cdot| \widetilde{y}_1 = y_1,\, \widetilde{y}_2 = y_2) 
=: 
\mathbb{P}(\theta \in \cdot| \widetilde{y}_{1:2} = y_{1:2}).
\]
This update models the Bayesian filtering from time point $t=1$ to $t=2$.
More generally, we can define a Bayesian filter as a map
 \[
 \{\mu_t, y_{t+1}\} 
 :=  
 \{\mathbb{P}(\theta \in \cdot| \widetilde{y}_{1:t} = y_{1:t}), y_{t+1}\}\mapsto \mathbb{P}(\theta \in \cdot| \widetilde{y}_{1:t+1} = y_{1:t+1})
  =:
  \mu_{t+1}.
 \]
Since $\mu_t$ reflects the knowledge about $\theta$ after seeing the data, this measure is called \emph{posterior (measure) at time $t$} for every $t \geq 1$.
The notions and explanations in \S \ref{Subs_Problem_Form} can be transferred to Bayesian filtering by replacing $\widehat{\theta}(y_{1:t})$ with $\mu_t$.

\subsection{Bayes' formula} \label{Subs_Bayes_Form}
The Bayesian learning procedure -- formalised by the map $\{\mu_t, y_{t+1}\} \mapsto \mu_{t+1}$ -- is fundamentally based on Bayes' formula.
In order to define Bayes' formula, we  make some simplifying, yet not necessarily restrictive assumptions.

In \S \ref{Subs_Problem_Form} the parameter space $X$ can be infinite-dimensional. 
In the remainder of this tutorial we consider $X:= \mathbb{R}^N$, a finite-dimensional parameter space.
For a treatment of the infinite-dimensional case we refer to \cite{Stuart2010}. 
Moreover, we assume that $\mu_0$ has a \emph{probability  density function (pdf)} $\pi_0: X \rightarrow \mathbb{R}$ with respect to the Lebesgue measure. 
This allows us to represent $\mu_0$ by $$\mu_0(A) := \int_A \pi_0(\theta) \mathrm{d}\theta,$$
for any measurable set $A \subseteq X$.
Furthermore, we assume that the \emph{model evidence}
\begin{align*}
{Z_{t+1}(y_{t+1})} 
&:= 
\mu_t\left(L_{t+1}(y_{t+1}|\cdot)\right)  
:= 
\int_X L_{t+1}(y_{t+1}|\theta)  \mathrm{d}\mu_t(\theta)\\ 
&:= 
\int_X L_{t+1}(y_{t+1}|\theta) \pi_t(\theta)  \mathrm{d}\theta, \quad t \geq 0,
\end{align*}
is strictly positive and finite. 
Then it follows that the conditional measures $\mu_1, \mu_2,\dots$ have pdfs $\pi_1, \pi_2,\dots$ on $X$ as well. 
The associated densities are given recursively by \emph{Bayes' formula}
\begin{equation}\label{Eq_Bayes_Formula}
\pi_{t+1}(\theta) 
= 
\frac{1}{Z_{t+1}(y_{t+1})} L_{t+1}(y_{t+1}|\theta) \pi_t(\theta),  \quad t \geq 0,  \quad \text{a.e. }\theta \in X.
\end{equation}
In some situations it is possible to use this formula to compute the densities $(\pi_t)_{t=1}^\infty$ analytically. In particular, this is possible if  $\pi_t$ is the pdf of a \emph{conjugate prior} for the likelihood $L_{t+1}$ for $t \geq 0$.
However, in the filtering of parameters of nonlinear models it is typically impossible to find conjugate priors. 
In this case we construct approximations to the densities $(\pi_t)_{t=1}^\infty$ and the measures $(\mu_t)_{t=1}^\infty$, respectively.
We discuss particle based approximations in \S \ref{Sec_SMC}.

Before moving on to the Bayesian formulation of the pendulum filtering problem, we briefly revisit Example \ref{Example_CondProb_Disc1} and explain Bayes' formula in this setting.
Note that Bayes' formula holds more generally for probability density functions that are given w.r.t. to $\sigma$-finite measures, for example, \emph{counting densities} on $\mathbb{Z}$; these are sometimes called \emph{probability mass functions (pmf)}.

\addtocounter{example}{-1}
\begin{example}[continued]
Recall that we consider a uniformly distributed random variable $\theta$ taking values in $\{1,\dots,10\}$. 
Hence the counting density is given by $\pi_0 \equiv 1/10$.
Furthermore, an oracle tells us that ``$\theta$ is about 4'' and we model this information by $\theta$ taking on values in $\{3,4,5\}$ with equal probability.
Hence
$$
L_1(y_1|k) 
= 
\mathbb{P}(\theta \text{ is about } 4 | \theta = k ) = \begin{cases}
1, &\text{ if } k = 3,4,5,\\
0, &\text{ otherwise.}
\end{cases}
$$
Now we compute the posterior counting density using Bayes' formula.
We arrive at
\begin{align*}
\pi_1(k) &= \frac{1}{Z_1(y_1)} \cdot \pi_0(k)  \cdot L_1(y_1|k) \\&=  \begin{cases}
\frac{1}{1\cdot\frac{1}{10}+1\cdot\frac{1}{10}+1\cdot\frac{1}{10}} \cdot \frac{1}{10} \cdot 1, &\text{ if } k = 3,4,5,\\
0, &\text{ otherwise.}
\end{cases}
\\
&=  \begin{cases}
1/3, &\text{ if } k = 3,4,5,\\
0, &\text{ otherwise.}
\end{cases}
\end{align*}
Since the values of counting densities are identical to the probability of the respective singleton, we obtain $\pi_1(k):= \mathbb{P}(\theta = k |\theta \text{ is about } 4)$.
This fits with the intuition discussed in Example \ref{Example_CondProb_Disc1}.
\end{example}

\subsection{Bayesian filtering formulation of the pendulum problem} \label{Subsec_Bayesian_Pendulum_formulation}
Next, we revisit the pendulum filtering problem in \S \ref{Sec_Pendulum}, and reformulate it as a Bayesian filtering problem. 
This requires us to define a prior measure for the parameter $g$.
The prior should include all information about the parameter before any observations are made.
The first piece of information about $g$ stems from the physical model which describes the motion of  the pendulum. 
Indeed, since we know that gravity attracts objects towards the center of the Earth, we conclude that the value of $g$ must be non-negative in the coordinate system we use. 
Furthermore, we assume that previous experiments and theoretical considerations tell us that $g \le 20m/s^2$, and that $g$ is probably close to the center of the interval $[0, 20]$. 
We model this information by a normal distribution with unit variance, centered at $10$, and truncated support on the interval $[0, 20]$.
Thus, the prior density is proportional to 
\[
 \pi_0(g) 
 \propto 
 \begin{cases} \exp\left(-\frac{1}{2}(g-10)^2\right), &0 \leq g \leq 20, \\
0, &\text{ otherwise.}
\end{cases}
\]
Furthermore, the data generating distribution in \S\ref{Sec_Pendulum} implies that the likelihoods are given by
\begin{equation}\label{Eq_Pendulum_Likelihoods}
  L_t(y_t|g)
  \propto \exp\left(-\frac{1}{2\sigma^2}\left(y_t - \mathcal{G}_t(g)\right)^2\right), \quad t=1,\dots,10,
  \end{equation}
  where $\sigma^2$ denotes the noise variance. 
In the numerical experiments we use the estimate $\sigma^2 = 0.0025$.
This is obtained by combining typical values for human reaction time and a forward Monte Carlo simulation as follows.
Studies (see e.g. \cite{thorpe1996speed}) suggest that a typical visual reaction time for humans is $450 \text{ms} \pm 100\text{ms}$. 
We now would like to use this information about the error in the time measurements for modelling the error in the angles. 
To this end we use the forward model parameterised by the mean of the prior distribution to compute a reference solution for a set of generated time measurements. 
We then compute a Monte Carlo estimate of the angle error resulting from adding a $N(0.45, 0.01)$ noise to the set of time measurements. 
The result of this numerical experiment indicates that $N(0, \sigma^2)$ with $\sigma^2 = 0.0025$ is a suitable model for the angle measurements error.
Ideally, one would estimate the noise variance directly in the pendulum experiment.

By inserting the prior density and the definition of the likelihoods into Bayes' formula \eqref{Eq_Bayes_Formula} we obtain a recursive expression for the densities $\pi_1, \ldots, \pi_{10}$ of the posterior measures $\mu_1, \ldots, \mu_{10}$ as follows:
\begin{align*}
\pi_{t+1}(g) 
&:= 
\frac{1}{Z_{t+1}(y_{t+1})} \exp\left(-\frac{1}{2\sigma^2}\left(y_{t+1} - \mathcal{G}_{t+1}(g)\right)^2\right) \pi_t(g), \ \ \text{for a.e. }g \in [0,20],\\
Z_{t+1}(y_{t+1}) &:= \mu_{t}\left(\exp\left(-\frac{1}{2\sigma^2}\left(y_{t+1} - \mathcal{G}_{t+1}(g)\right)^2\right)\right),   \ \ t \geq 0.
\end{align*}
 The family of posterior densities corresponding to the measurements in Table~\ref{Table_measurements} is depicted in Fig.~\ref{Fig_Im_Densities}.  
 As the time increases we see that the posterior densities become more concentrated.
 Note that a more concentrated density indicates less uncertainty compared to a flat density.
Thus, the picture is consistent with our intuition: The uncertainty in the parameter value for $g$ is reduced as more data becomes available.
We provide further comments on the estimation results in \S\ref{Subsubsec_SMC}.
\begin{figure}[t]
  \centering
  \includegraphics[scale=0.5]{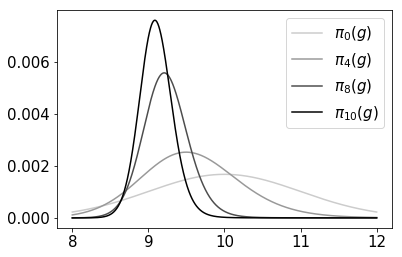}
  \caption{Sequence of posterior densities for a given value of $g$ for $t = 0\ (\text{prior}), 4, 8, 10$ data points. 
  The densities are computed with kernel density estimation based on a Sequential Monte Carlo particle approximation.
 }
  \label{Fig_Im_Densities}
\end{figure}

\section{Particle filters} \label{Sec_SMC}

It is typically impossible to compute the posterior measures $(\mu_t)_{t=1}^\infty$ analytically.
In this section we discuss two particle filters to approximate the solution of the Bayesian filtering problem. 
In \S \ref{Subs_SIS} we present the \emph{Sequential Importance Sampling} (SIS) algorithm. 
Unfortunately, SIS suffers from efficiency issues when used for filtering problems.
We explain this deficiency in a simple example.
It is possible to resolve the issues by extending SIS to the so-called \emph{Sequential Monte Carlo} (SMC) algorithm which we discuss in \S \ref{Subs_SMC}.

\begin{remark}
The sampling procedures presented in this section approximate posterior measures $(\mu_t)_{t= 1}^\infty$.
In addition, it is also possible to approximate integrals of the form $(\mu_t(f))_{t= 1}^\infty$, where $f: X \rightarrow \mathbb{R}$ is a $\mu_t$-integrable \emph{quantity of interest}.
To simplify the presentation we focus on the approximation of integrals rather than on the approximation of measures for the remaining part of this tutorial.
This can be done without loss of generality within the framework of weak representations of measures.
We outline this approach and give some examples for quantities of interest in Appendix A.
\end{remark}

\subsection{Sequential Importance Sampling}\label{Subs_SIS}

We are now interested in constructing an algorithm to approximate the sequence of posterior measures $(\mu_t)_{t=1}^\infty$ and thereby efficiently discretise the update rule of the learning process in Fig. \ref{Fig_tikz_learning}. 
To this end we consider \emph{particle-based} approximations.
These consist of a collection of $M$ weighted \emph{particles} (or \emph{samples}) $\{X_i, W_i\}_{i=1}^M$ where $X_i \in X$, $W_i \ge 0$ for $i = 1, \ldots, M$, and $\sum_{i=1}^MW_i = 1$. 
Additionally, the particles should be constructed in such a way that the \emph{random measure} $$\mu^M_t := \sum_{i=1}^MW_i\delta_{X_i}$$ converges \emph{weakly with probability one} to $\mu_t$ as $M \rightarrow \infty$. 
In particular, for any  bounded and continuous function $f : X \rightarrow \mathbb{R}$ (or bounded and Lipschitz-continuous function; see Prop. \ref{Prop_Convergence} in Appendix A),
\begin{align*}
  \mu_t^M(f) := \sum_{i=1}^MW_if(X_i) \xrightarrow{M \rightarrow\ \infty}  \mathbb{E}[f(\widehat\theta_t)] && \text{almost surely.}
\end{align*}
Note that $(\mu_t^M)_{M \geq 1}$ is a sequence of measure-valued random variables; that is a sequence of random variables defined on the space of measures. 
Almost every realisation of this sequence of measures converges weakly to a common deterministic measure as $M \rightarrow \infty$. 
More precisely, it holds that $$\mathbb{P}(\mu_t^M\rightarrow \mu_t, \text{ weakly, as }M \rightarrow \infty) = 1.$$

For simple distributions $\mu$ where direct sampling is possible an estimator with the desired properties above can be constructed by choosing independent random variables $X_1, \ldots, X_M$ distributed according to $\mu$ and by using uniform weights $W_i = 1/M$ for $i = 1, \ldots, M$. 
We denote this construction by the operator $S^M$ and call $\mu^M = S^M\mu$ the \emph{standard Monte Carlo} estimate of $\mu$. 
Since $S^M$ creates an empirical measure by using random variables the operator $S^M$ maps probability measures to random probability measures and is thus a non-deterministic operator.

However, it is typically impossible to sample from the posterior measures $(\mu_t)_{t=1}^\infty$ in the Bayesian filtering problem.
This precludes the construction of Monte Carlo estimates. 
Alternatively, we can use the \emph{Importance Sampling} (IS) method.
Let $\mu$, the \emph{target measure}, denote the probability measure to be estimated.
Let $\nu$ denote a probability measure from which it is possible to sample.
The measure $\nu$ is called \emph{importance measure}. 
Let $\mu$ be \emph{absolutely continuous} with respect to $\nu$.
This means that $\nu(A) = 0$ implies $\mu(A) = 0$ for any measurable set $A \subseteq X$.
Let furthermore $f$ denote a measurable, bounded function.
Then there exists a non-negative function $w : X \rightarrow \mathbb{R}$ such that
\begin{equation}\label{Eq_Radon_Nikodym}
  \mu(f) = \frac{\nu(f \cdot w)}{\nu(w)}.
\end{equation}
The intuition behind importance sampling is the following.
If the importance measure $\nu$ is close to the target measure $\mu$, then sampling from $\nu$ should be approximately equivalent to sampling from $\mu$. 

The IS estimate of $\mu$ is constructed by creating a Monte Carlo estimate of $\nu$.
Then, the Monte Carlo samples are reweighed using the expression in \eqref{Eq_Radon_Nikodym}.
This is necessary since we wish to obtain samples distributed according to the target measure $\mu$, and not samples distributed according to the importance measure $\nu$.
We arrive at
\begin{equation*}
  \mu^M(f)
  = 
  \frac{\nu^M(f \cdot w)}{\nu^M(w)} 
  = 
  \frac{\frac1M\sum_{i=1}^Mw(X_i)f(X_i)}{\frac1M\sum_{j=1}^Mw(X_j)} 
  :=
  \sum_{i=1}^MW_if(X_i).
\end{equation*}
In summary, importance sampling maps the Monte Carlo estimate $\nu^M$ to an updated estimate $\mu^M$ by adjusting the weights of the particles $X_1, \ldots, X_M$ according to the formula
\begin{equation*}
  W_i = \frac{w(X_i)}{\sum_{j=1}^Mw(X_j)}.
\end{equation*}

Now we apply the IS approximation in the context of the Bayesian filtering problem.
The prior measure $\mu_0$ is often tractable and direct sampling algorithms are available. 
It is then possible to create an initial Monte Carlo estimate of the prior measure.
Afterwards we can iteratively update the particles to incorporate the new knowledge from the observation $y_1, y_2, \ldots, y_T$ through importance sampling.
This follows the learning process described in \S \ref{Subs_Problem_Form}. 
Observe that Bayes' formula in \eqref{Eq_Bayes_Formula} tells us that each measure $\mu_{t+1}$ is absolutely continuous with respect to the previous measure.
We use this relation to define the nonlinear operator $\is_t$ for any probability measure $\mu$ over $X$ as follows:
\begin{equation*}
  (\is_{t+1} \mu)(f) 
  = 
  \frac{\mu(L_{t+1}(y_{t+1} | \cdot) \cdot f)}{\mu(L_{t+1}(y_{t+1} | \cdot))}.
\end{equation*}
Since $Z(y_{t+1}) = \mu_t(L_{t+1}(y_{t+1} | \cdot))$ for every $t \ge 0$ these operators can be used to describe reweighing in an importance sampling estimate with target measure $\mu_{t+1}$ and importance measure $\mu_t$. 
In particular, if for some $t \ge 0$ a particle approximation $\mu_t^M$ of $\mu_t$ is given by the particles $\{X^{(t)}_i, W^{(t)}_i\}$, the operator $\is_{t+1}$ can be used to define the following approximations:
\begin{align*}
  Z^M(y_{t+1}) &:= \mu_t^M(L_{t+1}(y_{t+1} | \cdot)) = \sum_{i=1}^MW^{(t)}_iL_{t+1}(y_{t+1} | X_i),\\
  \mu_{t+1}^M(f) &:= (\is_{t+1} \mu_t^M)(f) \\ &= \frac1{Z^M(y_{t+1})}\sum_{i=1}^MW^{(t)}_iL_{t+1}(y_{t+1} | X^{(t)}_i)f({X^{(t)}_i}) \\ &= \sum_{i=1}^MW^{(t+1)}_if({X^{(t)}_i}),
\end{align*}
where the particle weights $W^{(t+1)}_i$ are given by
\begin{equation*}
  W^{(t+1)}_i = \frac{W^{(t)}_iL_{t+1}(y_{t+1} | X^{(t)}_i)}{\sum_{j=1}^MW_j^{(t)}L_{t+1}(y_{t+1} | X_j^{(t)})}.
\end{equation*}
Note that the $\is_{t+1}$ update requires $\sum_{j=1}^MW_j^{(t)}L_{t+1}(y_{t+1} | X_j^{(t)}) > 0$. 
If the likelihood $L_{t+1}(y_{t+1} | \cdot )$ is strictly positive, then this condition is always satisfied.
We observe that the IS update formula only changes the weights $$(W_i^{(t)})_{i=1}^M \mapsto (W_i^{(t+1)})_{i=1}^M.$$ 
The positions $(X_i^{(t)})_{i=1}^M = (X_i^{(t+1)})_{i=1}^M$ of the particles remain unchanged.
Based on the recursive update formula, we construct an approximation to the sequence $(\mu_t)_{t=0}^\infty$ as follows:
\begin{align*}
\mu_0^M &= S^M\mu_0, \\
\mu_{t+1}^M &= \is_{t+1}(\mu_t^M).
\end{align*}

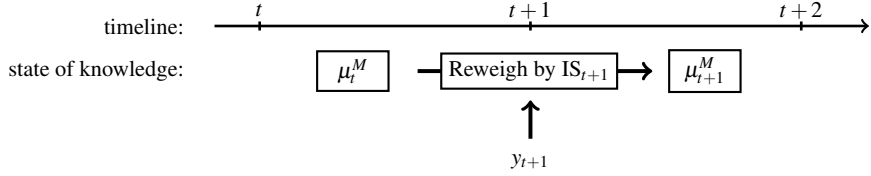
\begin{figure}[t]
\centering
\begin{tikzpicture}[scale=0.6]
  \draw	(1.5,0)node[anchor=east] {timeline:};
  \draw	(1.5,-1)node[anchor=east] {state of knowledge:};
  
  \draw	(3,0.35)node[align=center] {$t$};
  \draw	(9,0.35)node[align=center] {$t+1$};
  \draw	(15,0.35)node[align=center] {$t+2$};

  \draw[line width = 1.0, color = black,-](3,0.1) -- (3,-0.1);
  \draw[line width = 1.0, color = black,-](9,0.1) -- (9,-0.1);
  \draw[line width = 1.0, color = black,-](15,0.1) -- (15,-0.1);
  
  \draw[line width = 1.5, color = black,->](6.5,-1) -- (11.7,-1);
  \draw[line width = 1.0, color = black,->](2,0) -- (16.5,0);
  
  \draw	(4.25,-1)node[rectangle,anchor=west,text width=0.75cm, thick,draw=black,align=center,fill=white] {$\mu_t^M$};
  
  \draw[line width = 1.5, color = black,<-](9,-1.65) -- (9,-2.5);
  \draw	(9,-3)node[align=center] {$y_{t+1}$};

  \draw	(7,-1)node[align=center,rectangle,anchor=west,thick,draw=black,fill=white] {Reweigh by $\is_{t+1}$};

  \draw	(12.05,-1)node[align=center,rectangle,anchor=west,text width=0.75cm,thick,draw=black,fill=white] {$\mu_{t+1}^M$};
 \end{tikzpicture}
 \caption{Plot of the SIS approximation of the learning procedure. 
 The starting point is the approximation $\mu_t^M$ of the current state of knowledge at timepoint $t$. 
 At the timepoint $t+1$ we observe $y_{t+1}$ and use it to construct the importance sampling operator $\is_{t+1}$. 
 The new state of knowledge is approximated by $\mu_{t+1}^M$. }
\label{Fig_tikz_learning_sis}
\end{figure}
The sequential importance sampling algorithm is a natural and asymptotically correct approximation of the Bayesian filtering process, see e.g. \cite{Beskos2015,Bulte2018} for up to a finite number of observations. 
Now we discuss the accuracy of (sequential) importance sampling measured in terms of the so-called \emph{effective sample size (ESS)}.
Let $t \geq 0$. 
First, we define the constant $\rho_t > 0$ by
\begin{equation*}
  \rho_t := \frac{\mu_0(\mathbf{L}_t^2)}{\mu_0(\mathbf{L}_t)^2},
\end{equation*}
where $\mathbf{L}_t(\theta) = \prod_{i=1}^t L(y_t|\theta)$ is the joint likelihood of the dataset $y_{1:t}$.
One can show that the accuracy of (sequential) importance sampling up to time $t$ is equivalent to the accuracy of a standard Monte Carlo approximation with $M/\rho_t$ samples (see Remark \ref{Remark_ESS}).
The fraction $M/\rho_t$ is the \textit{effective sample size}. 
\begin{remark} \label{Remark_ESS}
 We sketch the derivation of the ESS in SIS for step $t \geq 0$. 
 We assume that we can sample $M'$ times independently from $\mu_t$. 
 Using these samples, we approximate $\mu_t$ by $S^{M'}(\mu_t)$, the standard Monte Carlo estimator.
 In this case, one can show that
 \begin{equation}\sup_{|f| \leq 1} \mathbb{E}\left(\Big|S^{M'}(\mu_t)(f) - \mu_t(f)\Big|^2\right) \leq \frac{4}{M'}. \label{EQ_MC_sample_size}
\end{equation}
This is an upper bound on the expected squared error between the integrals of any bounded test function $f$.
Hence, we can now use the \emph{sample size} $M'$ in equation \eqref{EQ_MC_sample_size} as an indicator for the accuracy of the standard Monte Carlo approximation $S^{M'}(\mu_t)$ of $\mu_t$. 
 In the SIS algorithm, $M$ samples are drawn from $\mu_0$ and are then reweighed by $\mathbf{L}_t$. 
 The resulting approximate measure $\mu_t^M$ fulfills the following error bound:
\begin{equation} \label{EQ_IS_sample_size}
\sup_{|f| \leq 1} \mathbb{E}\left(\Big|\mu_t^M(f) - \mu_t(f)\Big|^2\right) \leq \frac{4\rho_t}{M},
\end{equation}
see e.g. \cite{Agapiou2015}. 
Note that \eqref{EQ_MC_sample_size} can be derived from \eqref{EQ_IS_sample_size}, by setting $\mu_0 = \mu_t$ and $\mathbf{L}_t \equiv 1$, hence $\rho_t=1$.
To make the upper bounds in \eqref{EQ_MC_sample_size} and \eqref{EQ_IS_sample_size} comparable, we replace the sample size $M'$ in  \eqref{EQ_MC_sample_size} by the \textit{effective sample size} $M/\rho_t$ in \eqref{EQ_IS_sample_size}.
\end{remark}
In conclusion, we can interpret the ESS  as follows: 
\begin{itemize}
\item If $M/\rho_t \approx M$, then the estimation is nearly as accurate as sampling directly from the correct measure. 
Hence a value $\rho_t \approx 1$ is desirable.
\item If $M/\rho_t \ll M$, then the estimation is only as accurate as a Monte Carlo approximation with a very small number of particles.
\end{itemize}
Unfortunately, during the course of a filtering procedure, it may happen that $\rho_t$ explodes. 
We illustrate this by a simple example.
\begin{example}[Degeneracy] \label{Example_Gaussian_degeneracy}
Consider the problem of estimating the mean $m^\dagger$ of a normal distribution $\mathrm{N}(m^\dagger, 1)$ using $t \in \mathbb{N}$ samples $y_1, \ldots, y_t$ of the distribution. 
By choosing the conjugate prior $\mathrm{N}(0, 1)$ it is easy to see that the posterior distribution of the unknown parameter $m$ is equal to $\mathrm{N}\left(\frac{S_t}{t+1}, (1 + t)^{-1}\right)$, where $S_t = \sum_{i=1}^ty_i$. 
It is possible to compute $\rho_t$ analytically,
  \begin{equation*}
    \rho_t = \frac{t+1}{\sqrt{2t+1}}\exp\left(\frac{S_t^2}{(2t+1)(t+1)}\right).
  \end{equation*}
  Hence, $\rho_t = O(t^{1/2}; t \rightarrow \infty)$ grows unboundedly as $t$ increases.
  This implies that the effective sample size $M/\rho_t$ converges to $0$ as $t\rightarrow \infty$.  
\end{example}
In a filtering problem where $M/\rho_t \rightarrow 0$ as $t\rightarrow \infty$ we cannot use SIS since the estimation accuracy deteriorates over time.
It is possible, however, to estimate the effective sample size, and to use this estimate to improve SIS.

\begin{remark}
The SIS algorithm presented in \S\ref{Subs_SIS} is basic in the sense that the importance measure for $\mu_{t+1}$ at time point $t+1$ is simply the measure $\mu_{t}$ at time point $t$.
It is possible to construct alternative importance measures by using Markov kernels in the SIS algorithm.
We refer to the generic framework in \cite[\S 2]{DelMoral2006} for more details.
\end{remark}

\subsection{Sequential Monte Carlo}\label{Subs_SMC}
The effective sample size is a good indicator of the impoverishment of the particle estimate for $\mu_t$ due to the discrepancy between the prior and target distributions. 
In practice, we approximate $M/\rho_t$ by 
\begin{equation} \label{Eq_Ess_Approx_Sec_Moment}
  \ess 
  :=  
  M\frac{\mu_0^M(\mathbf{L}_t)^2}{\mu_0^M(\mathbf{L}_t^2)}
  =
  \frac{\left(\sum_{i=1}^M \mathbf{L}_{t}(X_i)\right)^2}{\sum_{i=1}^M \mathbf{L}_{t}(X_i)^2}.
\end{equation}
Note that $\ess$ is a consistent estimator of $M/\rho_t$.
Throughout the rest of this tutorial we will also refer to $\ess$ as \emph{effective sample size}.

If $\ess$ is small it is possible to discard particles situated in regions of the parameter space with low probability.
We can then replace these particles with particles that are more representative of the target distribution in the sense that they carry a larger weight.
This is done by introducing a \emph{resampling step} after adjusting the weights of the particles in the IS estimate. 
Precisely, we consider an approximation $\mu_t^M$ given by the particles $\{X_i^{(t)}, W_i^{(t)}\}_{i=1}^M$ at time $t\geq 0$. 
The weights of the approximation are updated by the $\is_{t+1}$ operator, and $\ess$ is computed using the updated set of weights $(W_i^{(t+1)})_{i=1}^M$. 
If $\ess$ is larger than a pre-defined threshold $\mlow \in (0, M]$, then the positions of the particles remain unchanged in the step from $\mu_t^M$ to $\mu_{t+1}^M$, giving $X_i^{(t+1)} = X_i^{(t)}, i=1,\dots,M$. 
If, on the other hand, $\ess \le \mlow$, then a new set of particles is sampled according to the updated weights.
The particle estimate is then given by $\mu_{t+1}^M = S^M(\is_{t+1}\mu_t^M)$ where we use the new particle set in the Monte Carlo estimate.
\begin{remark}
The choice of the threshold parameter $\mlow$ is highly problem dependent.
Doucet and Johansen  \cite{Doucet2011} mention that $\mlow = M/2$ is a typical choice. 
On the other hand, Beskos et. al. \cite{Beskos2015} use $\mlow = M$, that is, the resampling step is always carried out. 
Empirical tests targeting a small variance of the posterior mean estimator may also be helpful to define a suitable threshold.
\end{remark}
The resampling procedure can be performed as follows.
For $i=1,\dots,M$ sample $U_i \sim \mathrm{Unif}[0,1]$, a uniform random variable between $0$ and $1$.
Then, define
\begin{align*}
X_i^{(t+1)} &:= X_j^{(t)},\ \ j = \min\left\lbrace k \in \{1,\dots,M\} :\sum_{n=1}^k W_n^{(t+1)} \geq U_i \right\rbrace, \\
W_i^{(t+1)} &:= 1/M.
\end{align*}
Alternative ways of resampling are possible (see \cite{Gerber2017}).
Note that the resampling step can successfully eliminate particles in low density areas of the parameter space.
However, it still fails to reduce the particle degeneracy, since several particles may occur in the same position. 
Moreover, even  with resampling the particles cannot fully explore the parameter space, since their position remains fixed at all times. 

Due to the resampling in case of a small value $\ess$, or if $\ess$ is large, it is reasonable to assume that the remaining particles are approximately $\mu_{t+1}$-distributed.
The idea is now to scatter the particles in a way such that they explore the parameter space to reduce degeneracy without destroying the approximate $\mu_{t+1}$-distribution.
This can be achieved by a scattering with a $\mu_{t+1}$-invariant dynamic.
Such a dynamic is given by an ergodic Markov kernel
$K_{t+1} : X \times \mathcal{B}(X) \rightarrow [0, 1]$  that has $\mu_{t+1}$ as stationary distribution.
This means, if $X \sim \mu_{t+1}$ and $X' \sim K_{t+1}(X, \cdot)$, then $X' \sim \mu_{t+1}$.
Applying the Markov kernel repeatedly to a single particle will asymptotically produce particles that are $\mu_{t+1}$ distributed independently of the initial value (see \cite[Chapter~6]{Robert2004}).
Since we assumed that the particles are  approximately $\mu_{t+1}$-distributed, we typically rely only weakly on this asymptotic result.

\begin{figure}[t!]
  \centering
  \begin{tikzpicture}[scale=0.6]
  \draw	(1.5,0)node[anchor=east] {timeline:};
  
  \draw	(3,0.35)node[align=center] {$t$};
  \draw	(9,0.35)node[align=center] {$t+1$};
  \draw	(15,0.35)node[align=center] {$t+2$};

  \draw[line width = 1.0, color = black,-](3,0.1) -- (3,-0.1);
  \draw[line width = 1.0, color = black,-](9,0.1) -- (9,-0.1);
  \draw[line width = 1.0, color = black,-](15,0.1) -- (15,-0.1);
  \draw[line width = 1.0, color = black,->](2,0) -- (16.5,0);

  \draw[line width = 1.5, color = black,->](7,-0.85) -- (8.35,-0.85);
    \draw[line width = 1.5, color = black,->](9,-0.85) -- (10.35,-0.85);

  \draw	(6.45,-0.85)node[align=center,rectangle,anchor=west,thick,draw=black,fill=white,text width=0.6cm, text height=0.29cm] {$\is_{t+1}$};
  \draw	(8.45,-0.85)node[align=center,rectangle,anchor=west,thick,draw=black,fill=white,text width=0.6cm, text height=0.25cm] {$S^M$};
  \draw	(10.45,-0.85)node[align=center,rectangle,anchor=west,thick,draw=black,fill=white,text width=0.6cm, text height=0.29cm] {$K_{t+1}$};

  \draw	(1.5,-2.5)node[anchor=east] {state of knowledge:};
  \draw	(4.25,-2.5)node[rectangle,anchor=west,text width=0.75cm, thick,draw=black,align=center,fill=white] {$\mu_t^M$};
  
  \draw[line width = 1.5, color = black,->](6.5,-2.5) -- (11.7,-2.5);
  \draw	(7.6,-2.5)node[align=center,rectangle,anchor=west,thick,draw=black,fill=white] {SMC Update};

  \draw	(12.05,-2.5)node[align=center,rectangle,anchor=west,text width=0.75cm,thick,draw=black,fill=white] {$\mu_{t+1}^M$};

  \draw[line width = 1.5, color = black, dashed](7.65,-2.15) -- (6.5,-2.15) -- (6.5,-1.3);
  \draw[line width = 1.5, color = black, dashed](10.55,-2.15) -- (11.75,-2.15)-- (11.75,-1.3);
  
  \draw[line width = 1.5, color = black,<-](9,-3.15) -- (9,-4.1);
  \draw	(9,-4.5)node[align=center] {$y_{t+1}$};

\end{tikzpicture}

  \caption{Plot of the SMC approximation of the learning procedure. 
  The approximation is constructed in three steps: \emph{reweighing} to identify representative particles, \emph{resampling} to discard particles in low probability regions, and a \emph{correction} to adjust the position of the particles to the new state of knowledge. }
\label{Fig_tikz_learning_smc}
\end{figure}

The reader might find it not easy to construct an ergodic Markov kernel $K_t$ that has $\mu_t$ as a stationary measure, $t \geq 1$.
However, this is a standard task in so-called Markov Chain Monte Carlo methods (MCMC).
The literature on MCMC offers a large variety of suitable Markov kernels.
We mention Gibbs samplers, Hamiltonian Monte Carlo, Metropolis--Hastings, Random--Walk--Metropolis, Slice samplers (see e.g. \cite{Liu2004,Robert2004}).
In summary, we add a final step to the approximation, and apply once or several times the transition kernel $K_t$ to each of the particles to obtain a better coverage of the probability density of the posterior at time $t$. 
The resulting algorithm is the \emph{Sequential Monte Carlo} (SMC) method.
Its approximation of the learning process is depicted in Fig.~\ref{Fig_tikz_learning_smc}.

\section{Particle approximation of the Bayesian pendulum filtering} \label{Subs_Num_Pendulum}

We are now ready to compute an approximate solution to the Bayesian filtering problem in \S \ref{Subsec_Bayesian_Pendulum_formulation}. 
To this end we consider the likelihoods in \eqref{Eq_Pendulum_Likelihoods} and use the  set of measurements given in Table~\ref{Table_measurements}. 
Note that the distance between $\tau_6$ and $\tau_7$ is very short; we suspect that the timer has been operated twice instead of once.

%
\begin{table}[h]
  \begin{center}
    \begin{tabular}{|l r|cccccccccc|}
      \hline measurement & $t$ & $1$ & $2$ & $3$ & $4$ & $5$ & $6$ & $7$ & $8$ & $9$ & ${10}$  \\
      \hline
  time (in s) & $\tau_t$ &  1.51 & 4.06 & 7.06 & 9.90 & 12.66 & 15.40 & 15.58 & 18.56 & 21.38 & 24.36 \\ 
      \hline
    \end{tabular}
  \end{center}
  \caption{Time measurements corresponding to the angle $x(\tau_t)=0$ of a simple pendulum.}
  \label{Table_measurements}
\end{table}
The pendulum length is $\ell = 7.4m$.
The initial angle $x_0 = 5^\circ = \pi/36$, and the initial speed $v_0 = 0$. 
Since $x_0$ is small and $v_0 = 0$, it would be possible to consider the linearised IVP with an analytical solution (see Remark \ref{Remark_linearisedODE}).
In tests not reported here the filtering solution obtained with the analytical solution of the linearised IVP did not differ from the numerical solution of the nonlinear IVP.

\begin{lstlisting}[float,caption={Python file {\fontfamily{txtt}\selectfont smc\_approx.py}.},captionpos=b,label=lst:python]
import numpy as np
from scipy.stats import norm,truncnorm
from scipy.integrate import odeint
from particle_approximation import ParticleApproximation 

measurements = np.array([...]) # see Table 1.

# Define the probabilistic model.
prior = truncnorm(-10, 10, loc=10, scale=1)
error_model = norm(loc=0, scale=0.05)

# The IVP is defined by a RHS and initial values.
theta0 = [5*np.pi/180, 0]
def pendulum_rhs(theta, t, g):
    return [theta[1],-g/7.4 * np.sin(theta[0])]

# The forward response operator is discretized.
mesh = np.append(0, measurements)
def forward_response(g, n):
    sol = odeint(pendulum_rhs, theta0, mesh[:n+1], args=(g,))
    return sol[1:,0]

# The potential at time n defined using n data points.
def potential(g, n):
    if n == 0: return prior.logpdf(g)
    return error_model.logpdf(forward_response(g, n)).sum()
vec_potential = np.vectorize(potential)

# Define the proposal kernel for the correction steps.
def gaussian_proposal(x): return norm(loc=x, scale=.25).rvs()

# Approximate the sequence of posteriors.
approximation = ParticleApproximation(2500, prior)
for n in range(measurements.size):
    importance_potential = lambda x: vec_potential(x, n)
    target_potential = lambda x: vec_potential(x, n+1)
    approximation.smc_update(importance_potential, target_potential, gaussian_proposal, correction_steps=5, ess_ratio=1.3)
    mean = approximation.integrate(lambda x: x)
    var = approximation.integrate(np.square) - mean**2
    print("(posterior #%d) mean=%f, var=%f" % (n+1,mean,var))
\end{lstlisting}

The numerical results presented in the following sections have been computed with {\sc Python}.  
The associated Python code is available online, see \url{https://github.com/BayesianLearning/PenduSMC}.
A part of the code is printed in Display \ref{lst:python}.
It can be combined with the data set of time measurements in Table~\ref{Table_measurements} to reproduce our results, and to carry out further experiments.
The code in Display~\ref{lst:python} also summarises the complete modelling cycle of the pendulum filtering problem.
The class \texttt{ParticleApproximation} can be found in Appendix B.

\subsection{Sequential Importance Sampling} \label{Subs_num_SIS}
We first construct approximations of the learning process using the SIS algorithm as illustrated in Fig. \ref{Fig_tikz_learning_sis}. 
We employ different numbers of particles, $M = 2^4,2^5,\dots,2^{12}$, and perform 50 simulations for each value of $M$.
In Fig. \ref{Fig_Im_SIS} we present the results of the SIS approximations of the posterior distribution at time $t = 10$ depending on the number of particles $M$.
\begin{figure}[h]
  \begin{minipage}{.5\textwidth}
    \includegraphics[width=\linewidth]{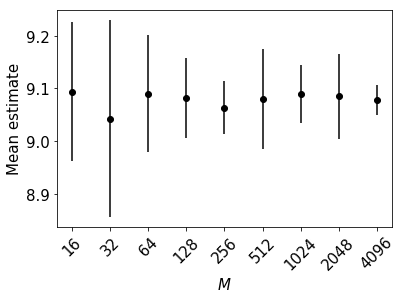}
  \end{minipage}
  \begin{minipage}{.5\textwidth}
    \includegraphics[width=\linewidth]{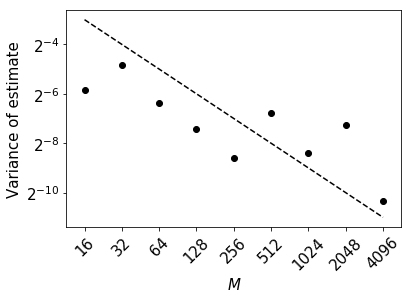}
  \end{minipage}
  \caption{Accuracy of the SIS approximation. 
  	In the plots we use data from runs with a variable number of particles indicated on the horizontal axis. We show results averaged over 50 runs per setting. Left: Plot displaying the convergence of the estimated mean of the posterior distribution. Each black dot represents the sample average over $50$ runs, and the bars represent the standard deviation of the mean within the runs. Right: Plot displaying the asymptotic convergence of the variance of the posterior mean estimate.}
  \label{Fig_Im_SIS}
\end{figure}
The left display in this figure shows that the variance of the posterior mean estimate is reduced as the number of particles $M$ is increased.
Moreover, we observe in the right display of the figure that the convergence rate of the variance of the posterior mean estimate is $O(1/M)$.
As expected this coincides with the convergence rate of a standard Monte Carlo approximation (see Remark \ref{Remark_ESS}).

Next, we fix the number of particles $M = 2500$, and investigate the accuracy in the SIS approximation by studying the effective sample size. 
Recall that in SIS only samples from the prior distribution are used throughout the filtering procedure. 
Anticipated by the plot of the densities in Fig. 3 we expect that these prior samples lead to a poor approximation of the sequence of posterior distributions constructed during the learning process.
This intuition is confirmed by numerical experiments shown in Fig. \ref{Fig_Im_SIS_Run}.
We see that the effective sample size is reduced dramatically in the first step of SIS with $t=1$.
In the final step with $t=10$ the effective sample size is only about $20\%$ of the initial sample size.
Unfortunately, a reduced effective sample size implies a loss of estimator accuracy.

\begin{figure}[h]
  \centering
  \begin{minipage}{.5\textwidth}
  \includegraphics[width=\linewidth]{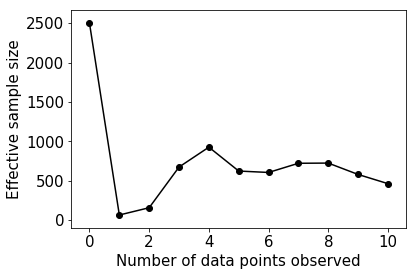}
  \end{minipage}
  \caption{Plot of the effective sample size for the SIS approximation of the learning process. }
  \label{Fig_Im_SIS_Run}
\end{figure}

Recall that the mean square error of a Monte Carlo estimator with $20\% \cdot M$ samples is 5 times larger compared to the mean square error with $M$ samples. 
In the simple pendulum setting, where the parameter space space is one-dimensional and compact, a sample size of $20\% \cdot M$ is likely still sufficient to obtain a useful posterior estimate. 
In real-world applications, however, we typically encounter high-dimensional and unbounded parameter spaces which require the exploration with a large number of representative samples.
In this case, the decrease of the effective sample size in SIS is a serious drawback. 
\begin{figure}[!t]
  \begin{minipage}{.5\textwidth}
    \includegraphics[width=\linewidth]{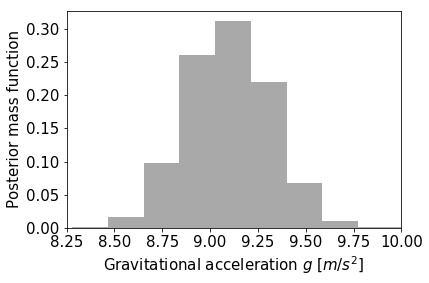}
  \end{minipage}
  \begin{minipage}{.5\textwidth}
    \includegraphics[width=\linewidth]{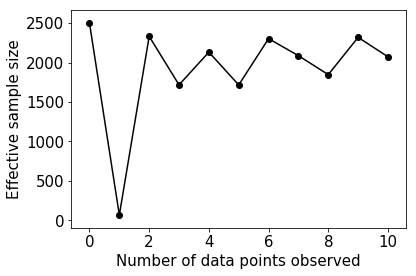}
  \end{minipage}
  \caption{SMC approximation of the final posterior distribution of the learning process at $t = 10$. 
  	Left: Histogram of the estimated measure $\mu_{10}$ computed by the particle approximation. 
  	Right: Plot of the effective sample size over time. 
  	Since $\mlow=1875$ the SMC algorithm performed $4$ resampling steps to maintain a large effective sample size.  }
  \label{Fig_Im_Typical}
\end{figure}

\subsection{Sequential Monte Carlo} \label{Subsubsec_SMC}
We proceed by constructing an approximation of the learning process using the SMC algorithm. 
The Markov kernel $K_t$ is given by a Random-Walk-Metropolis algorithm with target distribution $\mu_t$ and a Gaussian random walk with standard deviation $0.25$  as proposal distribution. 
Each Markov kernel $K_t$ is applied $5$ times to correct the approximated distribution.
We choose a minimal effective sample size corresponding to $75\%$ of the total population $M$.
This means that we resample if $\ess < \mlow := 75\% \cdot M$. 

We present a typical run of the algorithm in Fig. \ref{Fig_Im_Typical}.
Again $M = 2500$ particles are used to approximate the learning process.
In the right panel of Fig.~\ref{Fig_Im_Typical} we see that 4 resampling steps have been performed.
This maintains an ESS well over $1500$ and thus improves the accuracy of the posterior estimate with SMC compared to the estimate with SIS.

Moreover, in the left panel of Fig.~\ref{Fig_Im_Typical} we see that the SMC approximation to the posterior measure for $t=10$ is centered around the  value $g \approx 9.12 m/s^2$. 
To compare the estimate with the true value $g^\dagger = 9.808 m/s^2$ we compute the posterior probability of the event $\{|g - g^\dagger| < \varepsilon g^\dagger\}$.
This describes our posterior expectations about the closeness of the estimated parameter to the truth. 
We plot the results in Fig. \ref{Fig_eps}.
The posterior $\mu_{10}$ considers values close to the true parameter $g^\dagger$ unlikely. 
Aside from the first digit we would not be able to identify $g^\dagger$. 
However, the estimate is not bad given the very simple experimental setup.
The relative error associated with $g=9.12 m/s^2$ is 7\%; a similar result was obtained in \cite{Allmaras2013}. 

We suspect that the estimate for $g$ can be improved by using a more sophisticated physical model for the pendulum dynamics.
Moreover, the noise model could be improved. 
In tests not reported here we observed that the measurement error increases as time increases. 
This is plausible since the timer has been operated manually and not automatically by a sensor. 
Finally, we did not include the uncertainty in the initial condition $x_0$.

\begin{figure}[!t]
\centering
  \begin{minipage}{.5\textwidth}
    \includegraphics[width=\linewidth]{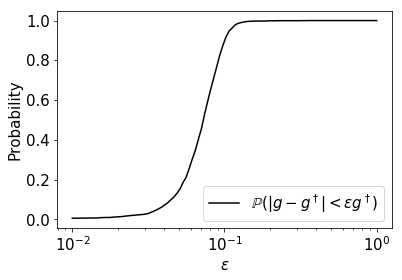}
\end{minipage}
  \caption{Comparison of the posterior measure $\mu^y$ with the true parameter $g^\dagger$. 
  We plot $\mu_{10}(|\cdot - g^\dagger| < \varepsilon g^\dagger)$ for various values of $\varepsilon$. 
  This value shows how likely the posterior measures sees the uncertain parameter $g$ to be $\varepsilon$-close to the true value $g^\dagger$. 
  The posterior probabilities are computed with SMC based on $2500$ particles.}
  \label{Fig_eps}
\end{figure}

To validate the SMC posterior estimate of $\mu_{10}$ with $M = 2500$ particles, we have performed a reference run using a Markov chain Monte Carlo sampler.
The MCMC sampler was run with 3000 posterior samples where the first 500 samples have been discarded to mitigate the burn-in effect. 
Using the MCMC and SMC samples we perform a kernel density estimation for the posterior pdf.
We present the estimation results in Fig. \ref{Fig_SMCMCMC}.
The posterior approximations do not differ significantly which leads us to conclude that the SMC approximation is consistent with the MCMC approximation.

\begin{figure}[!t]
\centering
  \begin{minipage}{.5\textwidth}
    \includegraphics[width=\linewidth]{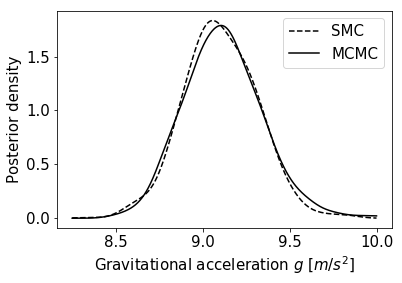}
  \end{minipage}
  \caption{Kernel density estimates of $\mu_{10}$ using the MCMC reference solution and the SMC solution. Each estimate is based on 2500 samples.}  
  \label{Fig_SMCMCMC}
\end{figure}

Moreover, we again investigate the convergence of SMC based on simulations with different numbers of particles  $M = 2^4,2^5,\dots,2^{12}$.
In each of these settings we consider the results of $50$ simulations.
We present the test results in Fig. \ref{Fig_Im_SMC}. 
We plot the posterior mean approximation at time $t = 10$ and the variance of the estimators within 50 simulations.
\begin{figure}[h]
  \begin{minipage}{.5\textwidth}
    \includegraphics[width=\linewidth]{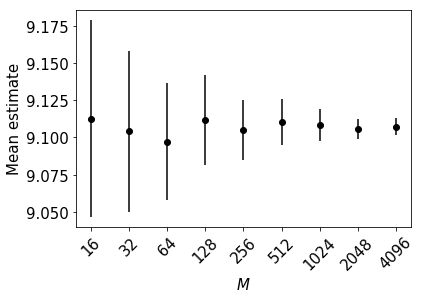}
  \end{minipage}
  \begin{minipage}{.5\textwidth}
    \includegraphics[width=\linewidth]{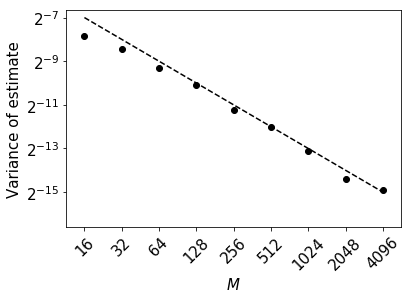}
  \end{minipage}
  \caption{Accuracy of the SMC approximation. 
  We use data from SMC runs with a variable number of particles indicated on the horizontal axis of each plot. 
  We show results averaged over 50 runs per setting.  
  Left: The convergence of the estimated mean of the posterior distribution. 
  The black dots represent the average over $50$ runs and the bars represent the standard deviation of the mean within the runs. 
  Right: The convergence of the variance of the posterior mean estimate.}
  \label{Fig_Im_SMC}
\end{figure}

Comparing Fig. \ref{Fig_Im_SIS} with Fig. \ref{Fig_Im_SMC} we see that the posterior mean estimates obtained with SMC differ only slightly from the estimates obtained with SIS.
Moreover, we observe again that the variance of the SMC estimate decreases with the rate ${O}(1/M)$ as $M$ increases (cf. Remark \ref{Remark_ESS}). 
However, we see a much higher variability in the posterior mean estimates in SIS.
In particular, we need a smaller number of particles in SMC compared to SIS to reach a certain variance of the estimator.
Recall, however, that the SMC algorithm involves parameters which need to be selected by the user; we mention the threshold sample size $\mlow$, and the Markov kernels $K_t$ together with the number of MCMC steps.
In contrast, the SIS algorithm does not require parameter tuning, and is simpler to implement.
Hence SIS could be used if the forward response operators $(\mathcal{G}_t)_{t \geq 1}$ can be evaluated cheaply; this allows more evaluations within a given computational budget to reach a desired accuracy.
SMC should be preferred if the evaluation of the forward response operators is computationally expensive, and if a small number of MCMC steps is sufficient to mix the particles efficiently in the parameter space. 

\subsection{Continuing to learn: data sets from further experiments}
\begin{figure}[h]
\centering
  \begin{minipage}{.5\textwidth}
    \includegraphics[width=\linewidth]{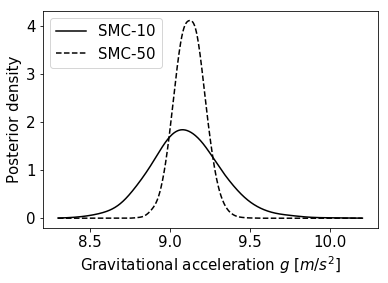}
  \end{minipage}
  \caption{Posterior density based on the 10 measurements given in Table \ref{Table_measurements} (SMC-10) and based on 40 further measurements (SMC-50). }
  \label{Fig_10vs50}
\end{figure}
The filtering problem we discussed so far is kept simple for demonstration purposes.
However, in reality one would perhaps not update the posterior after the collection of every single measurement (so every couple of seconds), and one would try to use more data sets, possibly from different sources.
Indeed, the data set in Table \ref{Table_measurements} is only one of several data sets that were collected in independent experiments, each carried out by different individuals.
A more realistic filtering problem is to update the posterior measure for $g$ after each round of measurements has been completed by an individual.
Then, the posterior measure reflects the knowledge obtained from  a fixed number of independently performed experiments.

Notably, the SIS and SMC algorithm can be used in this situation as well.
It is possible to process the measurements individually or in batches, with both SIS and SMC, without changes to the implementation.
Moreover, there is in principle no limit for the number of measurements that can be used in an update.
However, since we expect that the posterior measures associated with a larger number of measurements will be more concentrated, it is likely that the degeneracy of the effective sample size of SIS will become more pronounced.
For this reason we use SMC to include data sets from further pendulum experiments.

In Fig. \ref{Fig_10vs50} we plot the posterior density associated with a single experiment with 10 measurements (see the set-up in \S \ref{Subsubsec_SMC}) along with the posterior density for six experiments with a total number of 50 measurements; this includes the 10 measurements in the single experiment setting.
We observe that the posterior measure after six experiments is more concentrated (informed) than the posterior after one experiment.
Thus the uncertainty in the parameter $g$ is reduced after seeing data from different experiments.
This is consistent with our intuition about learning.

\section{Discussion}\label{sec:discussion}

We finish with some comments on current research directions.

For Bayesian filtering problems with finite-dimensional data and parameter spaces together with a simple mathematical model it is straightforward to prove the well-posedness of the Bayesian filtering problem (i.e. the existence, uniqueness, and stability of the posterior measure).
In general, we will have to consider infinite-dimensional parameter spaces (e.g. random fields), infinite-dimensional data spaces (e.g. the image of a tumor), and complex mathematical models.
The theoretical framework for the well-posedness proof has been established in \cite{Stuart2010}, however, the conditions therein need to be verified on a case-by-case basis.

A related problem is the inheritance of certain properties of measures (e.g. Gaussianity, convexity, tail behavior) within the sequence of posterior measures.
This is important since in a (time-dependent) filtering problem the posterior at time $t$ becomes the prior at time $t+1$.
Hence it is not sufficient to establish the existence of a posterior without studying its properties as well.

It is often possible to prove that the Bayesian learning process converges to a measure concentrated in a small neighbourhood around the true parameter in the large data limit (Bernstein-von-Mises theorem, Doob's consistency theorem, see \cite[\S 10]{vdVaart1998}).
However, the particle filters that approximate the learning process often suffer from a path degeneracy, meaning that a large number of updates decreases the estimator quality. 
We illustrated this for SIS where the effective sample size decreased over time. 
For SMC, however, a similar effect might occur, see \cite{Andrieu1999} or \cite[\S 5.1.2.]{Latz2018}, where path degeneracy is observed.

Finally, we mention that in current works on the SMC convergence the effect of the MCMC steps is often not considered (see e.g. \cite[Thm 3.1]{Beskos2015}) and requires further investigations.
If the MCMC steps are analysed, the results require rather strong assumptions, see \cite{Whiteley2013} and \cite[\S 5]{Beskos2015}.
In addition, the mechanism of importance sampling, which is a crucial component in SMC, is not fully analysed to date. 
Recent works \cite{SanzAlonso2018} establish bounds on the necessary sample size, however, it is unclear which metric is appropriate to study the nearness and, eventually, convergence of measures.

\section*{Appendix A: Quantities of interest and weak representations of measures} \label{Subs_Integrals}
For a sequence of measures $(\mu_t)_{t=1}^\infty$ we are typically interested in computing the expected value of a measurable and $\mu_t$-integrable ($t \geq 1$) \emph{quantity of interest} $f : X \rightarrow \mathbb{R}$ with respect to a measure in the sequence, that is $$\mu_t(f) := \int f \mathrm{d}\mu_t, \ \ t \geq 1.$$
\begin{example}[Quantities of interest] \label{EXAm_QoI}
Let $t \geq 1$.
\begin{itemize}
\item Let $f(x) := x_i$ denote the canonical projection onto the coordinate $x_i$ of $x$. 
Then, $\mu_t(f)$ is the \emph{mean} of the $i$th marginal density of $\mu_t$. 
Moreover, $\mu_t(f^k)$ is the $k$th moment ($k \in \mathbb{N}$) of the same marginal.
\item Let $A \subseteq X$ be measurable, and let $f$ be given by
$$
f(x) = \begin{cases} 1, &\text{if } x \in A,\\
0, &\text{otherwise.}
\end{cases}
$$
Then, $\mu_t(f) =: \mu_t(A)$ is the \emph{probability} that the parameter takes on values in $A$.
\end{itemize}
\end{example}

The quantities of interest in Example \ref{EXAm_QoI} are interesting in their own right.
More importantly, it is possible to use integrals with respect to certain functions to fully represent a measure. 
\begin{proposition}\label{criteria}
Let $\mu, \nu$ be two measures on $(X, \mathcal{B}X) := (\mathbb{R}^N, \mathcal{B}\mathbb{R}^N)$. Then, the identiy $\mu = \nu$ holds, if one of the following conditions is satisfied.
\begin{enumerate}
\item $\mu(A) = \nu(A)$ \ for all $A \in \mathcal{B}X$,
\item $\mu(f) = \nu(f)$ \ for all bounded, measurable functions $f: X\rightarrow \mathbb{R}$,
\item $\mu(f) = \nu(f)$ \ for all bounded, continuous functions $f: X\rightarrow \mathbb{R}$,
\item $\mu(f) = \nu(f)$ \ for all bounded, Lipschitz-continuous functions $f: X\rightarrow \mathbb{R}$.
\end{enumerate}
\end{proposition}
\begin{proof}
Condition (1.) is the definition of $\mu = \nu$. 
Condition (2.) implies (1.) by setting $$f(x) := \begin{cases}1, &\text{ if }x \in A\\ 0, &\text{ otherwise,}\end{cases}$$
for any $A \in \mathcal{B}X$.
Conditions (3.) and (4.) imply equivalence of the characteristic functions of $\mu, \nu$ which implies that $\mu = \nu$ \cite[Thm. 13.16, Thm. 15.8]{Klenke2014}.
\end{proof}
Moreover, we can use the criteria in Proposition~\ref{criteria} to investigate the convergence of a sequence of measures.

\begin{proposition} \label{Prop_Convergence}
Let $(\mu^M)_{M=1}^\infty$ be a sequence of measures and let $\nu$ denote a further measure on $(X, \mathcal{B}X) := (\mathbb{R}^N, \mathcal{B}\mathbb{R}^N)$. 
Then $\mu^M \rightarrow \nu$ as $M \rightarrow \infty$

\begin{enumerate}
\item in total variation, if one of the following conditions holds: 
\begin{enumerate}
\item $\mu^M(A) \rightarrow \nu(A)$ as $M \rightarrow \infty$ for all $A \subseteq X$ measurable,
\item $\mu^M(f) \rightarrow \nu(f)$ as $M \rightarrow \infty$ for all bounded, measurable functions $f: X\rightarrow \mathbb{R}$,
\end{enumerate}
\item  weakly, if it converges in total variation, or if one of the following conditions holds: 
\begin{enumerate}
\item $\mu^M(f) \rightarrow \nu(f)$ as $M \rightarrow \infty$ for all bounded, continuous functions $f: X\rightarrow \mathbb{R}$,
\item $\mu^M(f) \rightarrow \nu(f)$ as $M \rightarrow \infty$ for all bounded, Lipschitz-continuous functions $f: X\rightarrow \mathbb{R}$.
\end{enumerate}
\end{enumerate}
\end{proposition}
\begin{proof}
Convergence in total variation holds if $$\lim_{M\rightarrow \infty}\sup_{A \in \mathcal{B}X}|\mu^M(A) - \nu(A)| = 0.$$
This follows directly by condition (1.a) or (1.b).
(2.a) is the definition of weak convergence, (2.b) is implied by the Portmanteau Theorem \cite[Thm. 13.16]{Klenke2014}.
\end{proof}
Thus, instead of investigating the properties of measures directly, it is possible to study integrals of functions with respect to these measures.

\section*{Appendix B: Source code of the ParticleApproximation class}

\begin{lstlisting}[caption={Python file {\fontfamily{txtt}\selectfont particle\_approximation.py}.},captionpos=b]
import numpy as np
import matplotlib.pyplot as plt

class ParticleApproximation:
    def __init__(self, num_particles, prior, init=True):
        self.num_particles = num_particles
        self.prior = prior
        
        # create an initial Monte Carlo approximation of the prior
        if init:
            self.particles = prior.rvs(size=num_particles)
            self.weights = np.full(num_particles, 1.0/num_particles)
    
    @staticmethod
    def load(filename, prior):
        data = np.load(filename)
        approximation = ParticleApproximation(data['particles'].size, prior, init=False)
        approximation.particles = data['particles']
        approximation.weights = data['weights']
        return approximation
    
    def save(self, filename):
        np.savez(filename, particles=self.particles, weights=self.weights)
        
    def hist(self, **kwargs):
        plt.hist(self.particles, weights=self.weights, **kwargs)
             
    def integrate(self, f):
        fv = np.vectorize(f)
        return np.dot(fv(self.particles), self.weights)
    
    def sample(self, size):
        return np.random.choice(self.particles, size=size, p=self.weights, replace=True)
    
    def resample(self):
        self.particles = np.random.choice(self.particles, size=self.num_particles, p=self.weights, replace=True)
        self.weights = np.full(self.num_particles, 1/self.num_particles)
   
    def effective_sample_size(self):
        return 1 / np.dot(self.weights, self.weights)

    def reweight(self, importance_potential, target_potential):
        # Update is done in log-scale
        self.weights = np.log(self.weights)
        
        # Compute log-importance-weights and update current weights
        importance_weights = target_potential(self.particles) - importance_potential(self.particles)
        self.weights += importance_weights
        
        # Return to linear-scale to normalize weights
        self.weights = np.exp(self.weights)
        self.weights /= self.weights.sum()

    def mh_correction(self, target_potential, proposal_kernel, n_steps):
        total_accepted = 0
        
        # Sample from the proposal kernel, conditioned on currect particles
        proposals = proposal_kernel(self.particles)
        
        proposal_potentials = target_potential(proposals)
        current_potentials = target_potential(self.particles)
        
        for i in range(n_steps):
            # Compute the log acceptance ratio
            potential_ratio = proposal_potentials - current_potentials
            prior_ratio = self.prior.logpdf(proposals) - self.prior.logpdf(self.particles)
            acceptance_ratio = np.exp(potential_ratio + prior_ratio)

            # Randomly accept the transitions based on the log acceptance ratio
            accepted = np.random.uniform(size=self.num_particles) < acceptance_ratio
            self.particles[accepted] = proposals[accepted]
            
            total_accepted += np.sum(accepted)
            
            # Recompute necessary potentials for next step
            if i < n_steps - 1:
                # Update current potentials
                current_potentials[accepted] = proposal_potentials[accepted]
                
                # Sample new proposals
                proposals = proposal_kernel(self.particles)
                proposal_potentials = target_potential(proposals)

        return total_accepted / (n_steps * self.num_particles)
                
    def smc_update(self, importance_potential, target_potential, proposal_kernel, correction_steps, ess_ratio):
        self.reweight(importance_potential, target_potential)
        acceptance_ratio = self.mh_correction(target_potential, proposal_kernel, n_steps=correction_steps)

        ess = self.effective_sample_size()
        if ess < self.num_particles/ess_ratio:
            self.resample()
            
        return (acceptance_ratio, ess)
\end{lstlisting}

\begin{acknowledgement}
The authors thank the anonymous reviewer for pointing out an error in an earlier version of the implementation of SIS, and for several helpful comments.
JL and EU gratefully acknowledge the support by the German Research Foundation (DFG) and the Technische Universit\"at M\"unchen through the International Graduate School of Science and Engineering within project 10.02 BAYES.
The measurements were collected by the authors and with the much appreciated help by Elizabeth Bismut, Ionu\c{t}-Gabriel Farca\c{s}, Mario Teixeira Parente, and Laura Scarabosio during the Open House Day on 21 October 2017 at the Forschungszentrum Garching, Germany.
\end{acknowledgement}

\bibliographystyle{spmpsci}
\bibliography{library}

\begin{thebibliography}{10}
\providecommand{\url}[1]{{#1}}
\providecommand{\urlprefix}{URL }
\expandafter\ifx\csname urlstyle\endcsname\relax
  \providecommand{\doi}[1]{DOI~\discretionary{}{}{}#1}\else
  \providecommand{\doi}{DOI~\discretionary{}{}{}\begingroup
  \urlstyle{rm}\Url}\fi

\bibitem{Agapiou2015}
Agapiou, S., Papaspiliopoulos, O., Sanz-Alonso, D., Stuart, A.M.: Importance
  sampling: intrinsic dimension and computational cost.
\newblock Statist. Sci. \textbf{32}(3), 405--431 (2017)

\bibitem{Allmaras2013}
Allmaras, M., Bangerth, W., Linhart, J.M., Polanco, J., Wang, F., Wang, K.,
  Webster, J., Zedler, S.: {Estimating Parameters in Physical Models through
  Bayesian Inversion: A Complete Example}.
\newblock SIAM Review \textbf{55}(1), 149--167 (2013)

\bibitem{Andrieu1999}
{Andrieu}, C., {De Freitas}, N., {Doucet}, A.: {Sequential MCMC for Bayesian
  Model Selection}.
\newblock Proc. IEEE Workshop HOS  (1999)

\bibitem{Beskos2016}
Beskos, A., Jasra, A., Kantas, N., Thiery, A.: {On the convergence of adaptive
  sequential Monte Carlo methods}.
\newblock Ann. Appl. Probab. \textbf{26}(2), 1111--1146 (2016)

\bibitem{Beskos2017a}
Beskos, A., Jasra, A., Law, K.J.H., Tempone, R., Zhou, Y.: {Multilevel
  Sequential Monte Carlo Samplers}.
\newblock Stoch. Proc. App. \textbf{127}(5), 1417--1440 (2017)

\bibitem{Beskos2015}
Beskos, A., Jasra, A., Muzaffer, E.A., Stuart, A.M.: {Sequential Monte Carlo
  methods for Bayesian elliptic inverse problems}.
\newblock Stat. Comput. \textbf{25}(4), 727--737 (2015)

\bibitem{Bulte2018}
Bult\'{e}, M.: {Sequential Monte Carlo for time-dependent Bayesian inverse
  problems}.
\newblock Bachelor's thesis, Technische Universit\"at M\"unchen  (2018)

\bibitem{Collis2017}
Collis, J., Connor, A.J., Paczkowski, M., Kannan, P., Pitt-Francis, J., Byrne,
  H.M., Hubbard, M.E.: {Bayesian Calibration, Validation and Uncertainty
  Quantification for Predictive Modelling of Tumour Growth: A Tutorial}.
\newblock Bull. Math. Biol. \textbf{79}(4), 939--974 (2017)

\bibitem{Dashti2017}
Dashti, M., Stuart, A.M.: {The Bayesian Approach to Inverse Problems}.
\newblock In: R.~Ghanem, D.~Higdon, H.~Owhadi (eds.) Handbook of Uncertainty
  Quantification, pp. 311--428. Springer (2017)

\bibitem{DelMoral2004}
{Del Moral}, P.: {Feynman-Kac Formulae - Genealogical and Interacting Particle
  Systems with Applications}.
\newblock Springer (2004)

\bibitem{DelMoral2013}
{Del Moral}, P.: {Mean Field Simulation for Monte Carlo Integration}.
\newblock Chapman and Hall/CRC (2013)

\bibitem{DelMoral2006}
{Del Moral}, P., Doucet, A., Jasra, A.: {Sequential Monte Carlo samplers}.
\newblock J. R. Statist. Soc. B \textbf{68}(3), 411--436 (2006)

\bibitem{Doucet_Webpage}
Doucet, A.: {Sequential Monte Carlo Methods $\&$ Particle Filters Resources}.
\newblock \url{http://www.stats.ox.ac.uk/~doucet/smc_resources.html}.
\newblock Accessed: 2018-07-20

\bibitem{Doucet2011}
Doucet, A., Johansen, A.M.: A tutorial on particle filtering and smoothing:
  fifteen years later.
\newblock In: The {O}xford handbook of nonlinear filtering, pp. 656--704.
  Oxford Univ. Press, Oxford (2011)

\bibitem{Evensen2003}
Evensen, G.: {The Ensemble Kalman Filter: theoretical formulation and practical
  implementation}.
\newblock Ocean Dynamics \textbf{53}(4), 343--367 (2003)

\bibitem{Gerber2017}
{Gerber}, M., {Chopin}, N., {Whiteley}, N.: {Negative association, ordering and
  convergence of resampling methods}.
\newblock ArXiv e-prints  (2017)

\bibitem{Humpherys2012}
Humpherys, J., Redd, P., West, J.: {A Fresh Look at the Kalman Filter}.
\newblock SIAM Review \textbf{54}(4), 801--823 (2012)

\bibitem{Kahle2018}
{Kahle}, C., {Lam}, K.F., {Latz}, J., {Ullmann}, E.: {Bayesian parameter
  identification in Cahn-Hilliard models for biological growth}.
\newblock ArXiv e-prints  (2018)

\bibitem{Kaipio2005}
Kaipio, J., Somersalo, E.: {Statistical and Computational Inverse Problems}.
\newblock Springer (2005)

\bibitem{Kantas2014}
Kantas, N., Beskos, A., Jasra, A.: {Sequential Monte Carlo Methods for
  High-Dimensional Inverse Problems: A case study for the Navier-Stokes
  equations}.
\newblock SIAM/ASA J. Uncertain. Quantif. \textbf{2}(1), 464--489 (2014)

\bibitem{Klenke2014}
Klenke, A.: {Probability Theory: a comprehensive course}.
\newblock Springer (2014)

\bibitem{Latz2016}
Latz, J.: {Bayes Linear Methods for Inverse Problems}.
\newblock Master's thesis, University of Warwick (2016)

\bibitem{Latz2018}
Latz, J., Papaioannou, I., Ullmann, E.: {Multilevel Sequential$^2$ Monte Carlo
  for Bayesian inverse problems}.
\newblock {J. Comput. Phys.} \textbf{{368}}, {154 -- 178} ({2018})

\bibitem{Law2015}
Law, K., Stuart, A., Zygalakis, K.: Data assimilation. A mathematical
  introduction, \emph{Texts in Applied Mathematics}, vol.~62.
\newblock Springer, Cham (2015)

\bibitem{Lima2017}
Lima, E.A.B.F., Oden, J.T., Wohlmuth, B., Shahmoradi, A., Hormuth~II, D.A.,
  Yankeelov, T.E., Scarabosio, L., Horger, T.: {Selection and Validation of
  Predictive Models of Radiation Effects on Tumor Growth Based on Noninvasive
  Imaging Data}.
\newblock Comput. Methods Appl. Mech. Eng. \textbf{327}, 277--305 (2017)

\bibitem{Liu2004}
Liu, J.S.: {Monte Carlo Strategies in Scientific Computing}.
\newblock Springer (2004)

\bibitem{Nakamura2015}
Nakamura, G., Potthast, R.: Inverse Modeling.
\newblock 2053-2563. IOP Publishing (2015)

\bibitem{Neal2001}
Neal, R.M.: {Annealed importance sampling}.
\newblock Stat. Comp. \textbf{11}(2), 125--139 (2001)

\bibitem{Noce06}
Nocedal, J., Wright, S.J.: Numerical Optimization, second edn.
\newblock Springer, New York, NY, USA (2006)

\bibitem{Robert2007}
Robert, C.P.: {The Bayesian Choice}, 2nd edn.
\newblock Springer (2007)

\bibitem{Robert2004}
Robert, C.P., Casella, G.: {Monte Carlo Statistical Methods}.
\newblock Springer (2004)

\bibitem{SanzAlonso2018}
Sanz-Alonso, D.: Importance {S}ampling and {N}ecessary {S}ample {S}ize: {A}n
  {I}nformation {T}heory {A}pproach.
\newblock SIAM/ASA J. Uncertain. Quantif. \textbf{6}(2), 867--879 (2018)

\bibitem{Sarkka2013}
S\"arkk\"a, S.: Bayesian filtering and smoothing, \emph{Institute of
  Mathematical Statistics Textbooks}, vol.~3.
\newblock Cambridge University Press, Cambridge (2013)

\bibitem{Schillings2017}
Schillings, C., Stuart, A.M.: {Analysis of the Ensemble Kalman Filter for
  Inverse Problems}.
\newblock SIAM J. Numer. Anal. \textbf{55}(3), 1264--1290 (2017)

\bibitem{Stuart2010}
Stuart, A.M.: {Inverse problems: A Bayesian perspective}.
\newblock In: Acta Numerica, vol.~19, pp. 451--559. Cambridge University Press
  (2010)

\bibitem{thorpe1996speed}
Thorpe, S., Fize, D., Marlot, C.: Speed of processing in the human visual
  system.
\newblock Nature \textbf{381}(6582), 520 (1996)

\bibitem{vdVaart1998}
van~der Vaart, A.W.: Asymptotic statistics.
\newblock Cambridge Series in Statistical and Probabilistic Mathematics.
  Cambridge University Press (1998)

\bibitem{Whiteley2013}
Whiteley, N.: Sequential {M}onte {C}arlo samplers: error bounds and
  insensitivity to initial conditions.
\newblock Stochastic Analysis and Applications \textbf{30}(5), 774--798 (2013)

\end{thebibliography}

\end{document}